%% file: JumpIntoPast.tex
\documentclass{article}
\usepackage{arxiv}
\usepackage[T1]{fontenc}
\usepackage{amsmath}
\usepackage{amsfonts}
\usepackage{amsthm}
\usepackage{graphicx}

\usepackage{booktabs} 

\usepackage[ruled]{algorithm2e} 

\SetAlFnt{\small}
\SetAlCapFnt{\small}
\SetAlCapNameFnt{\small}
\SetAlCapHSkip{0pt}
\IncMargin{-\parindent}

\newtheorem{example}{Example}
\newtheorem{theorem}{Theorem}
\newtheorem{corollary}{Corollary}

\newcommand{\ds}{\displaystyle}
\newcommand{\NN}{\mathbb{N}}
\newcommand{\RR}{\mathbb{R}}

\newcommand{\BigO}[1]{\mathrm{O}\left(#1\right)}
\newcommand{\E}[1]{\mathrm{E}\left[#1\right]}
\newcommand{\Geo}[1]{\mathrm{Geo}(#1)}
\newcommand{\EXP}[1]{\mathrm{Exp}\left(#1\right)}
\newcommand{\Beta}[2]{\mathrm{B}(#1,#2)}

\begin{document}
\title{Taking snapshots from a stream}  
\author{Dominik Bojko,\\
 \texttt{Dominik.Bojko@pwr.edu.pl},
\AND
Jacek Cicho\'{n},\\
 \texttt{Jacek.Cichon@pwr.edu.pl},\\[1cm]
 Wroclaw University of Science and Technology,\\
 Faculty of Fundamental Problems of Technology,\\
 pl. Grunwaldzki 13, 50-377 Wroclaw, Poland}
\date{June 2022}


\maketitle

\begin{abstract}
This work is devoted to a certain class of probabilistic snapshots for elements of the observed data stream. We show you how one can control their probabilistic properties and we show some potential applications.  
Our solution can be used to store information from the observed history with limited memory. It can be used for both web server applications and Ad hoc networks and, for example, for automatic taking snapshots from video stream online of unknown size. 
\end{abstract}



\input{Introduction}
\input{General}

\section{Special cases}

In this section we will consider some examples of sequences $(\alpha_n)$ which may have applications for controlling the history of a massive streams of data.
The case $\alpha_n = \frac{1}{n}$ was analyzed in Example \ref{example:uniform} and the generated pointers (i.e. the values of the random variables $K_n$) may be treated as controllers of behavior of a stream at position $\frac{n}{2}$ after reading $n$ items. 
In a series of four subsections we shall analyze behavior of the random variable $K_n$ for sequences of the form 
$\alpha_n = \min\{1,\frac{g}{n^\alpha}\}$ for various fixed parameters $\alpha$. This is a summary of our results:
\begin{itemize}
\item if $\alpha = 0$ and $g\in(0,1)$ then the random variable $K_n$ converges in distribution to the geometric distribution $\mathrm{Geo}(g)$, hence $\E{K_n} \sim \frac{1}{g}$
\item if $\alpha\in(0,1)$ then the random variable $\frac{K_n}{n^\alpha}$ converges in distribution to exponential distribution $\mathrm{Exp}(g)$, hence $\E{K_n} \sim \frac{n^\alpha}{g}$
\item if $\alpha = 1$ then the random variable $\frac{K_n}{n}$ converges in distribution to the beta distribution $\mathrm{Beta}(1,g)$, hence $\E{K_n} \sim \frac{n}{g+1}$
\item if $\alpha \in (1,2)$ then $\E{L_n} = \frac{g}{2-\alpha} n^{2-\alpha}+\BigO{\max\{1,1\!\!1_{\{\frac{3}{2}\}}(\alpha)\ln{n},n^{3-2\alpha}\}}$
\item if $\alpha = 2$ then $\E{L_n} = g\ln(n) + \BigO{1}$ 
\item if $\alpha >2$ then $\E{L_n} = \BigO{1}$. 
\end{itemize}

\input{Fixed}

\input{SubLinear}
\input{Linear}

\input{SubQuadratic}
\input{Quadratic}
\input{SuperQuadratic}
\input{Applications}

\paragraph*{Acknowledgements}
The authors would like to thank Zbigniew Go\l{}\k{e}biewski and Jakub Lemiesz for usefull hints and comments on the proofs presented in this paper.
The work is supported by the Polish National Science Center grant number 2013/09/B/ST6/02258 .

\bibliographystyle{unsrtnat}
\bibliography{streaming}

\end{document}

%% file: Introduction.tex
\section{Introduction}
Suppose that we are observing a stream $x_1,x_2,\ldots$ of data. Our goal is to keep an element from this stream with a prescribed position. For example, we may want to keep the element $x_i$ with index $i$ close to $\left\lfloor n/2\right\rfloor$ after reading first $n$ elements from the stream. Of course this problem is trivial if we have a direct access to all elements $x_1,\ldots,x_n$. But in the case of large amount of data keeping all information into memory is expensive or undesirable. Suppose hence that our memory resources are limited.  

In this article we investigates series of randomized procedures which allow us to choose elements located near the required position in the stream of data. All these procedures are based on the same schema. They only differs on a sequence $(\alpha_n)$ of probabilities which are used for control of changes of stored data. In each cases we will have $\alpha_1=1$. We will use three variables: $K$, $n$ and $data$. Initially we put $K=0$, $n=0$ and we set $data$ as nil value:

\begin{verbatim}
Initialization: K:= 0; n:= 0; data:= nil;
\end{verbatim}
After reading an element $x$ from the stream we call the following update procedure:
\begin{verbatim}
procedure Update(x)
  n:= n+1
  if (random() <= a(n)) 
  begin
    K:= 1;
    data: = x
  end else begin
    K:= K+1
  end
end
\end{verbatim}
In the procedure Update we used the function random() which is high quality pseudo-random generator of random reals from the interval $[0,1]$. The function $a(n)$ represents the probabilities sequence $(\alpha_n)$. We call the variable $K$ a \textbf{probabilistic snapshot} from data stream.

\paragraph{Connection with reservoir sampling}
This class of algorithms may be treated as a subclass of reservoir sampling algorithms with reservoir of size 1 (see \cite{tille2006sampling}). 
Suppose that we use this procedure with the sequence $\alpha_n = \frac{1}{n}$. In this case we obtain the classical Jeffrey Vitter's Algorithm R published in \cite{Vitter:1985:RSR:3147.3165}. Let $K_n$ denotes value of the random variable $K$ after reading $n$ items from the stream. It is well known that  in this case (i.e. when $\alpha_n=\frac{1}{n}$) we have $\Pr[K_n=i] = \frac1n$ for each $i \in \{1,\ldots,n\}$, i.e. that the random variable $K_n$ has the uniform distribution on the set $\{1,\ldots,n\}$. 
Therefore $E[K_n] = (n+1)/2$. 

\paragraph{Applications}
Here we show some possible applications of methods discussed in Section \ref{sec:Applications}:
\begin{itemize}
\item The solution proposed in this paper may by used for storing fixed number of snapshots from an observed movie of unknown length. For example we may want to store short samples of the movie taken at times closed to $0$, $\frac{1}{10}T$, $\frac{2}{10}T$ \ldots $\frac{9}{10}T$, $T$ from a movie of length $T$.    
\item We may need a sample of data from times close to $n$, $n-C$, $n-2\cdot C$,\ldots, $n-k\cdot C$, where $n$ is an index of current item, $C$ is a fixed distance and $k$ is a reasonably small natural number. 
\item We may need to observe a sample from stock market in such a way that the snapshots from the past should be less rare than snapshots from times close to the present.
\end{itemize}

\subsection{Mathematical notations and background}
\begin{itemize}
\item{We denote by $\E{X}$ the expected value of the random variable $X$.}
\item{Let us recall that a discrete random variable $X$ has geometric distribution with parameter $p \in [0,1]$ ($X \sim \Geo{p}$)
if $P[X=k] = (1-p)^{k-1}p$ for $k\geq 1$. If $X \sim \Geo{p}$ then $\E{X} = \frac{1}{p}$.}
\item{A random variable $Y$ has an exponential distribution with parameter $g$ ($Y \sim \EXP{\lambda}$) if its support is $[0,\infty)$ and
$\Pr[Y>x] = e^{-\lambda x}$ for each $x\geq 0$. If $Y \sim \EXP{\lambda}$ then $\E{Y} = \frac{1}{\lambda}$.}
\item{A random variable $Z$ has the Beta distribution with parameters $a,b>0$ ($Z \sim \Beta{a}{b}$) if its support is $[0,1]$ and
its density is given by the function $f(x) = \frac{\Gamma(a+b)}{\Gamma(a)\Gamma(b)}x^{a-1}(1-x)^{b-1}$, where $\Gamma(x)$ is the standard generalization of the factorial function. 
If $Z \sim \Beta{a}{b}$ then $\E{Z} = \frac{a}{a+b}$.}
\item{A sequence $X_1, X_2, \ldots$ of real-valued random variables is said to converge in distribution if $\lim _{n\to \infty }F_{n}(x)=F(x)$,
for every number $x\in\RR$ at which the function $F$ is continuous, where $F_n$ and $F$ are the cumulative distribution functions of random variables $X_n$ and $X$, respectively.}
\item{We denote by $H_n$ the $n$-th harmonic number, i.e. $H_n = \sum_{i=1}^{n}\frac{1}{i}$. We will use the following well known approximation $H_n = \ln(n) + \gamma + \BigO{\frac{1}{n}}$, where  $\gamma \approx 0.577$ is the Euler - Mascheroni constant.}
\item{We denote by $H_{n,a}$ the $n$-th harmonic number in power $a$, i.e. $H_{n,a}=\sum_{i=1}^{n}\frac{1}{i^a}$. If $a\neq 1$ and positive, then the following formula is useful:
\begin{gather*}
H_{n,a}-H_{k,a}=\frac{n^{1-a}-k^{1-a}}{1-a} + f^{(a)}(k) + \BigO{1}~,
\end{gather*}
as $n\rightarrow\infty$, where $|f^{(a)}(k)|\leq k^{-a}$.}
\item{In this paper $\zeta$ is Riemann's Zeta function, defined as $\zeta(z)=\sum_{n=1}^{\infty}n^{-z}$ (for $Re(z)>1$).}
\item{By $\mathrm{Ei}(x)$ we denote Exponential Integral, defined as $-\int_{-x}^{\infty}\frac{e^{-t}dt}{t}$. By $mathrm{E}_1(x)$ we denote the First Exponential Integral, defined as $\int_{1}^{\infty}\frac{e^{-xt}dt}{t}$. These functions have foregoing property $\mathrm{Ei}(x)=-\mathrm{E}_1(-x)$.}
\item{We will use the following inequalities (extension of the classical Weierstarass product inequality):
$$
  1 - \sum_{i=1}^{n} x_i \leq \prod_{i=1}^{n}(1-x_i)  < 1 - \sum_{i=1}^{n} x_i + \sum_{1\leq i< j \leq n} x_i x_j.
$$
which holds for any sequence $x_1, \ldots, x_n$ of real numbers from the interval $[0,1]$ (see \cite{10.2307/2688388}).}
\end{itemize}

%% file: General.tex
\section{General Properties}

Let us fix sequence $\alpha_n$ of reals from interval $[0,1]$ such that $\alpha_1=1$.
Let us consider the sequence $K_n$ of consecutive integers tracing the value of the variable $K$ during the run of the Algorithm 1 with controlling sequence $\alpha_n$. Clearly we have $1\leq K_n \leq n$ for each $n$.

Notice that $\Pr[K_1=1]=1$, $\Pr[K_n = 1] = \alpha_n$ and
\begin{equation}
\label{eq:RecurrenceForKn}
\Pr[K_{n+1} = k+1] = (1-\alpha_n)\Pr[K_n=k]
\end{equation}
The following formula may be derived from the previous one
\begin{equation}
\label{eq:RecurrenceForKn1}
  \Pr[K_n =k] = \alpha_{n-k+1} \prod_{i=n-k+2}^{n}(1-\alpha_i)~.
\end{equation}
Moreover
\begin{equation}
\label{eq:RecurrenceForKn2}
 \Pr[K_n\geq k] = \prod_{i=n-k+2}^{n} (1-\alpha_i)
\end{equation}

\begin{example} 
\label{example:uniform}
Let us consider the case when $\alpha_n = \frac{1}{n}$. We claim that $\Pr[K_n=k] = \frac1n$ for each $n$ and $k\in\{1,\ldots,n\}$. Indeed, $\Pr[K_1=1] = 1$. Let us assume that $\Pr[K_n=k] = \frac1n$ for each $k$. 
Then $\Pr[K_{n+1}=1] = \frac{1}{n+1}$ and
$$
\Pr[K_{n+1} = k] = (1-\alpha_{n+1}) \Pr[K_n=k-1] = \left(1- \frac{1}{n+1}\right) \frac1n = \frac{n}{n+1}\frac1n = \frac{1}{n+1}~,
$$
for $k \in \{2,\ldots,n+1\}$.
Notice that we just reconstructed the mentioned result from \cite{Alon:1999:SCA:310293.310314}: the random variable $K_n$ in this case is uniformly distributed on the set $\{1,\ldots,n\}$. Therefore $\E{K_n}= \frac{1+n}{2}$, so the expected value of $K_n$ is precisely the middle of the set $\{1,\ldots,n\}$. 
\end{example}

The sequence $\E{K_n}$ satisfies simple linear first order difference equation: 

\begin{theorem}
$\E{K_{n+1}} = 1 + (1-\alpha_{n+1}) \E{K_n}$
\end{theorem}

\begin{proof} Let $\beta_n = 1-\alpha_{n}$. The proof follows directly from the recurrence (\ref{eq:RecurrenceForKn}):
\begin{gather*}
\E{K_{n+1}} = 
\sum_{k=1}^{n+1} k \cdot p^{n+1}_k = 
\alpha_n + \sum_{k=2}^{n+1} k \cdot p^{n+1}_k = 
\alpha_n + \sum_{k=2}^{n+1} k \cdot (1-\alpha_{n+1})p^{n}_{k-1} = \\
\alpha_n + \beta_{n+1}\sum_{k=1}^{n} (k+1) \cdot p^{n}_{k} =
\alpha_n + \beta_{n+1}\sum_{k=1}^{n} k \cdot p^{n}_{k} + \beta_{n+1}\sum_{k=1}^{n} \cdot p^{n}_{k} = \\
\alpha_n + \beta_{n+1}\sum_{k=1}^{n} k \cdot p^{n}_{k} + \beta_{n+1} = 1 + (1-\alpha_{n+1})\E{K_n}~.
\end{gather*}
\end{proof}

We shall also use the variable $L_n = (n+1) - K_n$. From formula (\ref{eq:RecurrenceForKn1}) we immediately get
\begin{equation}
\label{eq:FormulaForLn}
\Pr[L_n = k] = \alpha_k \prod_{i=0}^{n-k-1}(1-\alpha_{n-i})
\end{equation}

\input{Monotonicity}

%% file: Monotonicity.tex
\subsection{Monotonicity of expected value, related to $\alpha_n$}
\label{sec:mon}
Let's consider two sequences of probabilities of a saving of $n$-th data: $(\alpha_n)$ and $(\alpha'_n)$, one majorized by another, in particular: $(\forall n\in\NN) \quad 0\leq \alpha_n \leq \alpha'_n\leq 1$.
Random processes $(K_n)_{n\in\NN}$ and $(K'_n)_{n\in\NN}$ are corresponding to sequences $(\alpha_n)$ and $(\alpha'_n)$ respectively, f.e. $P(K_n=1)=\alpha_n$ and $P(K'_n=1)=\alpha'_n$.
We show that then $\E{K_n} \geq \E{K'_n}$.
The proof is a simple induction --- $\E{K_1} = \E{K'_1 = 1}$ for any sequences of probabilities.
Now, assume that  $\E{K_n} \geq \E{K'_n}$ for some $n\in\NN$.
Then $\E{K_{n+1}}=1+(1-\alpha_{n+1})\E{K_n}\geq 1+(1-\alpha'_{n+1})\E{K'_n}=\E{K'_{n+1}}$.

%% file: Fixed.tex
\subsection{Fixed value}

Let us fix a real number $a > 1$ and let $\alpha_n = \frac{1}{a}$ for each $n>1$. In this case we have a closed formula for $E[K_n]$,
namely, from Eq. (\ref{eq:RecurrenceForKn1}) we immediately get
$$
\Pr[K_n = k ] = \begin{cases} \frac{1}{a} (1-\frac{1}{a})^{k-1} &: 1\leq k<n\\ (1-\frac{1}{a})^{n-1} &: k=n\end{cases}
$$
From this formula we deduce that
\begin{enumerate}
\item The sequence $(K_n)$ of random variables converges in distribution to the geometrical distribution with parameter $\frac1a$.
\item $E[K_n] = a \left(1-\left(\frac{a-1}{a}\right)^n\right)$
\end{enumerate}
 
Therefore the generated pointer may be used for controlling a behavior of a stream at a fixed position in the past.

%% file: SubLinear.tex
\subsection{Sublinear Case}
\label{sec:sublinear}

In this section we consider the case when $\alpha_n = \min\{1,\frac{g}{n^\alpha}\}$ for some fixed $g>0$ and $\alpha \in (0,1)$.
We shall show that the normalized random variable $\frac{K_n}{n^{\alpha}}$ converges in distribution to the exponential $\EXP{g}$ distribution.
\begin{theorem} 
\label{sublinear:thm01}
Let $g>0$ and $\alpha\in(0,1)$. 
Let $\alpha_n = \min\{1,\frac{g}{n^\alpha}\}$ and let $x>0$. Then 
$$\lim_{n\to\infty} \Pr\left[\frac{K_n}{n^\alpha} \leq x\right] = 1- e^{-gx} ~.$$
\end{theorem}

\begin{proof}
Let $k=\lfloor n^\alpha x\rfloor$. Using formula (\ref{eq:RecurrenceForKn2}) we obtain
\begin{gather*}
\Pr\left[\frac{K_n}{n^\alpha} >x\right] = \Pr[K_n > n^\alpha x] = \Pr[K_n > k] = \prod_{i=n-k+1}^{n}\left(1-\frac{g}{i^\alpha}\right)
\end{gather*} 
(for sufficiently large $n$). Therefore
\begin{gather*}
\label{ineq01}
\Pr\left[\frac{K_n}{n^\alpha} >x\right] < \left(1- \frac{g}{n^\alpha}\right)^{k} \leq \left(1- \frac{g}{n^\alpha}\right)^{n^\alpha x - 1} 
\end{gather*} 
and
\begin{gather*}
\Pr\left[\frac{K_n}{n^\alpha} >x\right] > \left(1- \frac{g}{(n-k+1)^\alpha}\right)^{k} >
\left(1- \frac{g}{(n - \lfloor n^\alpha x\rfloor )^\alpha}\right)^{n^\alpha x} = \\
\left(1- \frac{g/(1-\lfloor n^\alpha x\rfloor/n)}{n^\alpha}\right)^{n^\alpha x} ~,
\end{gather*} 
so we see that both bounds on $\Pr\left[\frac{K_n}{n^\alpha} >x\right]$ converges to the same limit $e^{-gx}$ when $n$ tends to infinity.
\end{proof}

\begin{theorem}
If $g>0$, $\alpha\in(0,1)$ and  $\alpha_n = \min\{1,\frac{g}{n^\alpha}\}$ then $\lim_{n\to\infty} \E{\frac{K_n}{n^\alpha}} = \frac{1}{g}$.
\end{theorem}

\begin{proof}
From Theorem \ref{sublinear:thm01} and Fatou's Lemma we get $\frac{1}{g} \leq \liminf_{n\to\infty} \E{\frac{K_n}{n^\alpha}}$.
On the other hand, inequality (\ref{ineq01}) bears (for sufficiently large $n>g^{1/\alpha}$):
$$
\E{\frac{K_n}{n^{\alpha}}}=\int\limits_{0}^{\infty}\Pr\left[\frac{K_n}{n^{\alpha}}>t \right] dt\leq \int\limits_{0}^{\infty}\left(1-\frac{g}{n^{\alpha}}\right)^{n^{\alpha}t-1}dt=-\frac{1}{(n^{\alpha}-g)\ln(1-\frac{g}{n^{\alpha}})}\stackrel{n\rightarrow\infty}{\longrightarrow}\frac{1}{g}~.
$$
\end{proof}

\begin{corollary}
If $g>0$ and  $\alpha_n = \min\{1,\frac{g}{n}\}$ then $\E{K_n} = \frac{n^\alpha}{g} + \mathrm{o}(n^\alpha)$.
\end{corollary}

\paragraph{Remark} We know a more precise results in some special cases. For example, if $\alpha_n = \frac{1}{\sqrt{n}}$ then
$\E{K_n} = \sqrt{n} - \frac12 +\frac{1}{2\sqrt{n}} + \frac{1}{8n} +\BigO{\frac{1}{n^{3/2}}}$.

%% file: Linear.tex
\subsection{Linear Case}
\label{sec:linear}

In this section we consider the case when $\alpha_n = \min\{1,\frac{g}{n}\}$ for some fixed $g>0$.
We shall show that the normalized random variable $\frac{K_n}{n}$ converges in distribution to the Beta(1,g) distribution.

\begin{theorem} Let $g>0$. Let $\alpha_n = \min\{1,\frac{g}{n}\}$ and let $x\in (0,1)$. Then 
$$\lim_{n\to\infty} \Pr\left[\frac{K_n}{n} \leq x\right] = 1 - (1-x)^g ~.$$
\end{theorem}

\begin{proof}
Let $k=\lfloor nx\rfloor$. Then we have
\begin{gather*}
\Pr\left[\frac{K_n}{n} >x\right] = \Pr[K_n > n x] = \Pr[K_n > k] = \prod_{i=n-k+1}^{n}(1-\alpha_i)~.
\end{gather*} 
Therefore, for sufficiently large $n$, we have
\begin{gather*}
\Pr\left[\frac{K_n}{n} >x\right] = \prod_{i=n-k+1}^{n}\left(1-\frac{g}{i}\right) ~.
\end{gather*} 
Hence
\begin{gather*}
\ln\left(\Pr\left[\frac{K_n}{n} >x\right]\right)  = -\sum_{i=n-k+1}^{n}\ln\left(\frac{1}{1-\frac{g}{i}}\right) = 
- \sum_{i=n-k+1}^{n} \sum_{a\geq 1} \frac{1}{a} \left(\frac{g}{i}\right)^a = \\
- g\sum_{i=n-k+1}^{n} \frac{1}{i} - \sum_{a\geq 2}\frac{g^a}{a} \sum_{i=n-k+1}^{n} \frac{1}{i^a}~.
\end{gather*}
Notice that $\sum_{i=n-k+1}^{n} \frac{1}{i} = H_{n}-H_{n-k}$, and
$$
-g(H_{n}-H_{n-k}) = -g\left(\ln(n)-\ln(n-k) + \BigO{\frac{1}{n}} + \BigO{\frac{1}{n-\lfloor nx\rfloor}}\right) = 
\ln\left(\frac{n-\lfloor nx\rfloor}{n}\right)^g +O\left(\frac{1}{n}\right)
$$
Therefore $\lim_{n\to\infty}(-g(H_{n}-H_{n-\lfloor nx\rfloor})) = \ln((1-x)^g)$. Let
$A_{n,k} = \sum_{a\geq 2}\frac{g^a}{a} \sum_{i=n-k+1}^{n} \frac{1}{i^a}$.
Then
\begin{gather*}
0<A_{n,k} < \sum_{a\geq 2} g^a \sum_{i=n-k+1}^{\infty} \frac{1}{i^a} = 
\sum_{i=n-k+1}^{\infty}\sum_{a\geq 2}  \frac{g^a}{i^a} =
\sum_{i=n-k+1}^{\infty} \frac{g^2}{i^2} \frac{1}{1 - \frac{g}{i}} < 2g^2 \sum_{i=n-k+1}^{\infty} \frac{1}{i^2}
\end{gather*}
(the last inequality holds for $n> \frac{2g}{1-x}$), hence
$$
A_{n,k} < 2 g^2 \sum_{i=n-k+1}^{\infty} \frac{1}{i(i-1)} = \frac{2g^2}{n-k} = \frac{2g^2}{n-\lfloor nx\rfloor} = \BigO{\frac{1}{n}}~. 
$$
Thus
$$
\ln\left(\Pr\left[\frac{K_n}{n} >x\right]\right) = \ln\left(1-\frac{\lfloor nx\rfloor}{n}\right)^g + \BigO{\frac{1}{n}}.
$$
Therefore 
$$
\Pr\left[\frac{K_n}{n} >x\right] = \left(1-\frac{\lfloor nx\rfloor}{n}\right)^g e^{\BigO{\frac{1}{n}}} = 
\left(1-\frac{\lfloor nx\rfloor}{n}\right)^g +\BigO{\frac{1}{n}} ~, 
$$
so $\lim_{n\to\infty}\Pr\left[\frac{K_n}{n} >x\right] = (1-x)^g$.
\end{proof}

\begin{corollary}
If $g>0$ and  $\alpha_n = \min\{1,\frac{g}{n}\}$ then $\lim_{n\to\infty} \E{\frac{K_n}{n}} = \frac{1}{g+1}$.
\end{corollary}

\begin{proof}
Notice that $\frac{K_n}{n} \leq 1$, hence the sequence $(\frac{K_n}{n})_{n\in\NN}$ is bounded, therefore the convergence in distribution of the sequence $(\frac{K_n}{n})_{n\in\NN}$ to a random variable $Y$ with Beta(1,g) distribution implies the convergence of moments (see e.g. \cite{Billingsley}), hence
$$
\lim_n \E{\frac{K_n}{n}} = \int_{0}^{1}x \frac{d}{dx}\left(1-(1-x)^g\right) dx = \frac{1}{1+g}~.
$$
 
\end{proof}

\begin{corollary}
If $g>0$ and  $\alpha_n = \min\{1,\frac{g}{n}\}$ then $\E{K_n} = \frac{n}{g+1} + \mathrm{o}(n)$.
\end{corollary}

\paragraph{Remark} A slightly more complicated calculus shows that in the case considered in this section we have
$\E{K_n} = \frac{n+1}{g+1} + \BigO{\frac{1}{n^g}}$.

%% file: SubQuadratic.tex
\subsection{Subquadratic Case}
In this section we consider the case when $\alpha_n = \min\{1,\frac{g}{n^\alpha}\}$ for some fixed $g>0$ and $\alpha \in (1,2)$. We introduce an auxilliary random variable $L_n:=n+1 - K_n$.
Let's define $g_0=\left\lceil g^{\frac{1}{a}}\right\rceil-1$.
Realize that $\alpha_k=\frac{g}{k^{\alpha}}$ for $k\geq g_0 +1$. Thus we briefly see that $Pr\left[L_n=k\right]=0$, whenever $k \leq g_0$ and
\begin{gather*}
\Pr\left[L_n=k\right]=
\ds{\frac{g}{k^{\alpha}}\prod\limits_{i=0}^{n-k-1}\left(1-\frac{g}{(n-i)^{\alpha}}\right)}~
\end{gather*}
otherwise.\\
In foregoing calculations let $k > g_0$.
At commencement we approach a lower bound, using standard Weierstrass' product inequality (ref??): 
\begin{gather}
\label{ineq01}
\Pr\left[L_n=k\right]\geq\frac{g}{k^{\alpha}}\left(1-\sum\limits_{i=k+1}^{n}\frac{g}{i^{\alpha}}\right)=\frac{g}{k^{\alpha}}\left(1-g(H_{n,\alpha}-H_{k,\alpha})\right)~.
\end{gather}
To find the upper limitation, we apply more precise extension of Weierstrass' product inequality (ref??): 
\begin{gather}
\Pr\left[L_n=k\right]\leq\frac{g}{k^{\alpha}}\left(1-\sum\limits_{i=k+1}^{n}\frac{g}{i^{\alpha}}+\sum_{i=k+1}^{n-1}\sum_{j=i+1}^{n}\frac{g^2}{(ij)^{\alpha}}\right)\leq\frac{g}{k^{\alpha}}\left(1-\sum\limits_{i=k+1}^{n}\frac{g}{i^{\alpha}} +\frac{1}{2}\sum_{i=k+1}^{n}\sum_{j=k+1}^{n}\frac{g^2}{(ij)^{\alpha}}\right)~.
\end{gather}
Hence, by repeating the same trick as in (\ref{ineq01}) we attain:
\begin{gather*}
\Pr\left[L_n=k\right]\leq\frac{g}{k^{\alpha}}\left(1-g(H_{n,\alpha}-H_{k,\alpha}) +\frac{g^2}{2}(H_{n,\alpha}-H_{k,\alpha})^2\right)~.
\end{gather*}

We would like to find close confinements for an expected value of $L_n$.\\
We consider the limitation which is easier to analyze ---
\begin{align*}
\E{L_n}&\geq \sum_{k=g_0+1}^{n}\frac{g}{k^{\alpha-1}}\left(1-g(H_{n,\alpha}-H_{k,\alpha})\right)=g(H_{n,\alpha-1}-H_{g_0,\alpha-1})- g^2 \sum_{k=1}^{n}\frac{\frac{k^{1-\alpha}-n^{1-\alpha}}{\alpha-1}+f^{(a)}(k)+\BigO{1}}{k^{\alpha-1}}\\
&\geq g\frac{n^{2-\alpha}}{2-\alpha}-g^2\left[\frac{(H_{n,2\alpha-2}-H_{g_0,2\alpha-2})}{\alpha-1}-\frac{n^{1-\alpha}}{\alpha-1}(H_{n,\alpha-1}-H_{g_0,\alpha-1})\right]+\BigO{1}\geq\\
&\geq \frac{gn^{2-\alpha}}{2-\alpha}-g^2\left[-\frac{n^{3-2\alpha}}{(\alpha-1)(2-\alpha)}+\frac{(H_{n,2\alpha-2}-H_{g_0,2\alpha-2})}{\alpha-1}\right]+\BigO{1}=\\
&=\left\{\begin{array}{ll}
\frac{gn^{2-\alpha}}{2-\alpha}-\frac{g^2n^{3-2\alpha}}{(3-2\alpha)(\alpha-1)}+\frac{g^2n^{3-2\alpha}}{(\alpha-1)(2-\alpha)} +\BigO{1}=\frac{gn^{2-\alpha}}{2-\alpha}-\frac{g^2n^{3-2\alpha}}{(3-2\alpha)(2-\alpha)}+\BigO{1}&;\quad\text{for } \alpha<\frac{3}{2},\\
\frac{gn^{2-\alpha}}{2-\alpha}-\frac{g^2 H_n}{\alpha-1} +\BigO{1}&;\quad\text{for } \alpha=\frac{3}{2},\\
\frac{gn^{2-\alpha}}{2-\alpha}+\BigO{1}&;\quad\text{for } \alpha>\frac{3}{2},
\end{array}\right.
\end{align*}
as $n\rightarrow\infty$.\\
The second one can be attained the same way:
$$
\E{L_n} \leq
\left\{\begin{array}{ll}
\frac{gn^{2-a}}{2-a}-\frac{g^2n^{3-2a}}{(3-2a)(2-a)}+\frac{g^2n^{4-3a}}{(4-3a)(3-2a)(2-a)}+\BigO{1}&;\quad\text{for } a<\frac{4}{3},\\
\frac{gn^{2-a}}{2-a}-\frac{g^2n^{3-2a}}{(3-2a)(2-a)}+\frac{g^3 H_n}{2(a-1)^2} +\BigO{1}&;\quad\text{for } a=\frac{4}{3},\\
\frac{gn^{2-a}}{2-a}-\frac{g^2n^{3-2a}}{(3-2a)(2-a)}+\BigO{1}&;\quad\text{for } \frac{4}{3}<a<\frac{3}{2},\\
\frac{gn^{2-a}}{2-a}-\frac{g^2 H_n}{a-1} +\BigO{1}&;\quad\text{for } a=\frac{3}{2},\\
\frac{gn^{2-a}}{2-a}+\BigO{1}&;\quad\text{for } a>\frac{3}{2},
\end{array}\right.
$$
as \(n\rightarrow\infty\).\\
Resuming up all the cases, we obtain:
\begin{corollary}
If $g>0$ is fixed, $\alpha\in(1,2)$ and $\alpha_n=\min\{1,\frac{g}{n^{\alpha}}\}$, then
\begin{gather*}
\E{L_n}=\left\{\begin{array}{ll}
\frac{gn^{2-\alpha}}{2-\alpha}-\frac{g^2n^{3-2\alpha}}{(3-2\alpha)(2-\alpha)}+\BigO{\max\{\ln n,n^{4-3\alpha}\}}&;\quad\text{for } \alpha<\frac{3}{2},\\
\frac{gn^{2-\alpha}}{2-\alpha}-\frac{g^2 H_n}{\alpha-1} +\BigO{1}&;\quad\text{for } \alpha=\frac{3}{2},\\
\frac{gn^{2-\alpha}}{2-\alpha}+\BigO{1}&;\quad\text{for } \alpha>\frac{3}{2},
\end{array}\right.
\end{gather*}
as \(n\rightarrow\infty\).
\end{corollary}
\paragraph{Remark}
A better generalization of Weierstrass' product inequality is known (ref??).
Notice that it could be helpful to get a more precise asymptotics for $\E{L_n}$, when $1<\alpha\leq\frac{4}{3}$. Nevertheless, it won't be thorough due to the approximations of the generalized harmonic numbers.
Unfortunately this calculation is very unpleasant.

%% file: Quadratic.tex
\subsection{Quadratic Case}
\label{sec:quadratic}
In this section we consider the case when $\alpha_n = \min\{1,\frac{g}{n^2}\}$ for some fixed $g>0$.
For computational purpose let introduce some auxiliary symbols: $g_0:=\sqrt{g}$ and $h:=\lceil g_0 \rceil$. We analyze random variable given by $L_n:=n+1-K_n$. Our main goal is to achieve an approximation for its expected value.
It's easy to realize that $\Pr\left[L_n\geq k\right]=1$, when $k<h$. Otherwise, we use (\ref{eq:RecurrenceForKn2}) to reach $\Pr\left[L_n\geq k\right]=1-\prod\limits_{i=k}^{n}\left(1-\frac{g}{i^2}\right)$.

\begin{theorem}
Let $g>0$ and  $\alpha_n = \min\{1,\frac{g}{n^2}\}$. Then $\E{L_n}=g\ln{n}+O(1)$.
\end{theorem}
\begin{proof}
Commence with some simple facts. The first uses the formula for expected value for non-negative discrete random variables:
\begin{gather}
\label{eq:ExpectedValueLn}
\E{L_n}=\sum\limits_{k=1}^{n}\Pr\left[L_n\geq k\right]=h-1+\sum\limits_{k=h}^{n}\left(1-\frac{g}{k^2}\right)~.
\end{gather}
The next two apply standard integral inequalities to find some confinements for sums:
\begin{gather}
\label{ineq:secondHarmonic}
-((n+1)^{-1} - k^{-1})=\int\limits_{k}^{n+1}x^{-2}dx\leq\sum\limits_{i=k}^{n}k^{-2}\leq \int\limits_{k-1}^{n}x^{-2}dx=-(n^{-1} - (k-1)^{-1})~
\end{gather}
and
\begin{gather}
\label{ineq:exponentialSum}
\sum\limits_{i=k}^{n}e^{-\frac{g}{k}}\leq \int\limits_{h}^{n+1}e^{-\frac{g}{x}}dx~.
\end{gather}
In the last formula we received integral, which isn't elementary function.
In order to use it, we need some additional properties of related known integrals.
It's known that 
$$
\int e^{-\frac{1}{x}}dx = \mathrm{Ei}\left(-\frac{1}{x}\right)+ x e^{-\frac{1}{x}}~,
$$
so 
\begin{gather}
\label{eq:ExponentialIntegral}
\int\limits_{h}^{n+1}e^{-\frac{g}{x}}dx=g\int\limits_{\frac{h}{g}}^{\frac{n+1}{g}}e^{-\frac{1}{x}}dx=g\left(\mathrm{Ei}\left(-\frac{1}{x}\right)+ x e^{-\frac{1}{x}}\right)\bigg|_{\frac{h}{g}}^{\frac{n+1}{g}}~.
\end{gather}
We also adduce the series expansion of exponential integral
\begin{gather}
\label{eq:seriesExpansion}
\mathrm{E}_1(x)=-\gamma-\ln{x} -\sum\limits_{k=1}^{\infty}\frac{(-1)^k x^k}{k! k}~.
\end{gather}

Now, we find an upper bound for $\Pr\left[L_n\geq k\right]$ (for $k\geq h$), using standard Weierstrass' inequality and (\ref{ineq:secondHarmonic}):
$$
1-\prod\limits_{i=k}^{n}\left(1-\frac{g}{i^2}\right)\leq 1-\left(1-\sum\limits_{i=k}^{n}\frac{g}{i^2}\right)\leq g\left(\frac{1}{k-1}-\frac{1}{n}\right)~.
$$
From above formula and (\ref{eq:ExpectedValueLn}), we obtain:
$$
\E{L_n}\leq h-1+g\sum\limits_{k=h}^{n}\left(\frac{1}{k-1}-\frac{1}{n}\right)=gH_n - gH_{h-2} +h-1-g+\BigO{n^{-1}}=g\ln n + \BigO{1}~.
$$

The second limitation is more difficult.
After removing the remainder of Taylor's formula for the function $e^x$, the inequality (\ref{ineq:secondHarmonic}) provides:
$$
\Pr\left[L_n\geq k\right]\geq 1-\prod\limits_{i=k}^{n}\left(1-\frac{g}{i^2}\right)\geq 1-\prod\limits_{i=k}^{n}e^{-\frac{g}{i^2}}=1-e^{-g\sum\limits_{i=k}^{n}\frac{1}{i^2}}\geq e^{\frac{g}{n+1}-\frac{g}{k}}~.
$$
Then we use the above part, the equality (\ref{eq:ExpectedValueLn}) once again and finally
(\ref{ineq:exponentialSum}) to attain:
$$
\E{L_n}\geq h-1+\sum\limits_{k=h}^{n}\left(1-e^{\frac{g}{n+1}-\frac{g}{k}}\right)=n-e^{\frac{g}{n+1}}\sum\limits_{k=h}^{n} e^{-\frac{g}{k}}\geq n-e^{\frac{g}{n+1}}\int\limits_{h}^{n+1}e^{-\frac{g}{x}}dx~.
$$
Remind that $\mathrm{Ei}(-y)=-\mathrm{E}_1(y)$ and by proceeding with (\ref{eq:ExponentialIntegral}) and (\ref{eq:seriesExpansion}) we procure:
\begin{align*}
\E{L_n}&\geq n-ge^{\frac{g}{n+1}}\left[-\mathrm{E}_1\left(\frac{g}{n+1}\right)+\mathrm{E}_1\left(\frac{g}{h}\right)+\frac{n+1}{g}e^{-\frac{g}{n+1}}-\frac{h}{g}e^{-\frac{g}{h}}\right]\\
&=n-ge^{\frac{g}{n+1}}\left[\ln\left(\frac{g}{n+1}\right)-\ln\left(\frac{g}{h}\right)+\frac{n+1}{g}e^{-\frac{g}{n+1}}-\frac{h}{g}e^{-\frac{g}{h}}+c_g+\BigO{\frac{1}{n}}\right]\\
&=-1+g\left(1+\BigO{\frac{1}{n}}\right)(\ln(n+1)-\ln(h))-gc_g+he^{-\frac{g}{h}}+\BigO{\frac{1}{n}}=g\ln n + \BigO{\frac{1}{n}}~,
\end{align*}
where $c_g=\sum_{k=1}^{\infty}\frac{(-1)^k \left(\frac{g}{h}\right)^k}{k! k}$.
\end{proof}
\paragraph{Remark}
It's possible to find neat limitations for $c_g$, because it's the alternating series. If we denote $n$-th partial sum of this series by $S_n$, then $c_g$ is between $S_r$ and $S_{r+1}$, where $r$ is a natural number, which satisfies $\frac{g}{h}^r<r! r$.
\paragraph{Remark}
The case $a=2$ is recognized as some kind of exception (due to appearance of $H_n$). We noticed that in fact, some pattern for expected value of $L_n$ extrapolates from the case, when $1<a<2$ on $a=2$.

%% file: SuperQuadratic.tex
\subsection{Superquadratic Case}

Finally, in this section we consider the case when $\alpha_n = \min\{1,\frac{g}{n^\alpha}\}$ for some fixed $g>0$ and $\alpha>2$. This case is the least interesting in our purpose, but we concisely consider $a>2$ for completeness. Again, we apply random variable $L_n:=n+1 - K_n$. The foregoing obvious fact is fulfilled:
\begin{gather*}
\Pr\left[L_n=k\right]=\alpha_k\prod\limits_{i=0}^{n-k-1}(1-\alpha_{n-i})\leq \alpha_k=\min\left\{1,\dfrac{g}{k^{\alpha}}\right\}~.
\end{gather*}
The foregoing formula provides an easy majorization for expected value:
\begin{gather*}
\E{L_n}\leq \sum_{k=1}^{\left\lfloor g^\frac{1}{\alpha} \right\rfloor}k + gH_{n,\alpha-1}<\frac{g^{\frac{1}{\alpha}}+1}{2}g^{\frac{1}{\alpha}} + g\zeta(\alpha-1)<\infty~.
\end{gather*}
\begin{corollary}
If $\alpha>2$, $g>0$ and $\alpha_n = \min\{1,\frac{g}{n^\alpha}\}$, then $\E{K_n}=n-\BigO{1}$.
\end{corollary}
\paragraph{Remark}
$L_n$ is a number of update of algorithm which saved the last snapshot. We should expect that our algorithm save only data from the very beginning. In this case, the algorithm is redundant, so we don't consider this case anymore.
Nevertheless this is quite interesting mathematical problem to study.

%% file: Applications.tex
\section{Applications}
\label{sec:Applications}

In this chapter we will discuss three examples of applications discussed in the previous chapters of probabilistic pointers. 

\subsection{Linear sampling}

Let us assume that we are observing data stream and that after reading $n$-th item we would like to have access to elements laying near points $\{\frac{k}{10}\cdot n: k=0,\ldots,10\}$. Of course, there is no problem with the element laying near $\frac{0}{n} \cdot n$ - it is sufficient to store the first element from the stream. As an element laying near $\frac{10}{10}\cdot n$ we may take the current item.
So we must propose some mechanism for dealing with remaining $9$ points.
     
Let us consider a series $K_n^{1}, \ldots K_n^{M}$ of independent random variables generated by the sequence $\alpha_n = \frac{1}{n}$. 
We know (see \ref{sec:linear}) that this is a sequence of independent random variables uniformly distributed in the set $\{1,\ldots,n\}$. 
Let $X_n= \{K_n^{1}, \ldots K_n^{M}\}$.
Let us fix some $0<\varepsilon < \frac{1}{10}$. Let $I_k = \{a\in\NN: |\frac{a}{n} - \frac{k}{10}| <\varepsilon\}$. 
Let us observe that $\Pr[I_k \cap X_n = \emptyset] \approx (1- 2 \varepsilon)^M$. This approximation is accurate for large $n$. So, for simplicity we shall assume that we have an equality. Therefore
$$
\Pr[\bigvee_{k=1}^{9} (I_k \cap X_n = \emptyset)] \leq 9\cdot (1- 2 \varepsilon)^M~. 
$$
The solution of inequality $9\cdot (1- 2 \varepsilon)^M \leq \eta$ is given by 
$M \geq \frac{\log \left(\frac{\eta }{9}\right)}{\log (1-2 \epsilon )}$. Putting in this formula $\varepsilon = \frac{1}{100}$ and
$\eta = 10^{-10}$ we get $M \geq 1248.5$. Therefore, if we take $M=1250$ snapshots then
$$
\Pr[\bigwedge_{k=1}^{9} (I_k \cap X_n \neq \emptyset)] > 1- \frac{1}{10^{10}}~.
$$
Hence, with a very high probability for each $k\in\{1,\ldots,10\}$ we are able to choose a point from the set $X_n$ which approximates $\frac{k n}{10}$ with precision $1\%$.  

We perform numerical experiments with this a collection of 1250 probabilistic snapshots 
$\mathcal{K}_n= (K_n^{1}, \ldots K_n^{1250})$ evolving independently according to the sequence $\alpha_n = \frac{1}{n}$. In the experiment whose results are shown in Fig. \ref{fig:LinearQuality} after each call to the procedure Update we calculated the quality of set of snapshots defined as
$$ 
Q(\mathcal{K}_n) = \frac{1}{n}\max\left\{\min\left\{\left|\frac{kn}{10}-K_n^j\right|: j=1,\ldots,1250\right\}: k=1,\ldots,9\right\}~.
$$  
\begin{figure}[th]
\centering
\includegraphics[width=0.75\columnwidth]{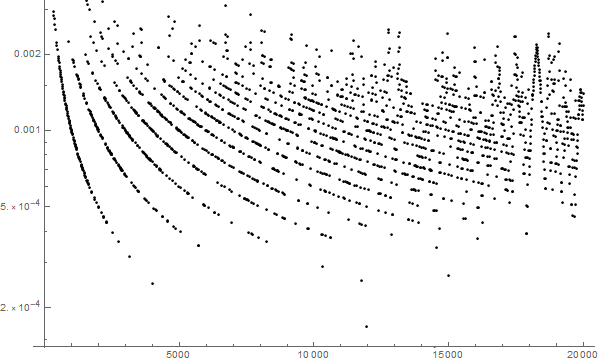}
\caption{Quality of approximation of points $\frac{kn}{10}$, $k = 1, \ldots, 9$,  by a collection of 1250 snapshots.}
\label{fig:LinearQuality}
\end{figure} 
We may observe in this figure several regularities which are connected with the specific method of generation of sequences $\mathcal{K}_n$ (for example, it is clear that sequences $\mathcal{K}_n$ and $\mathcal{K}_{n+1}$ are not independent of each other).  We can see that the predicted $1\%$ precision has been achieved.

\subsection{Bitcoin capitalization}
\begin{figure}[th]
\centering
\includegraphics[width=0.75\columnwidth]{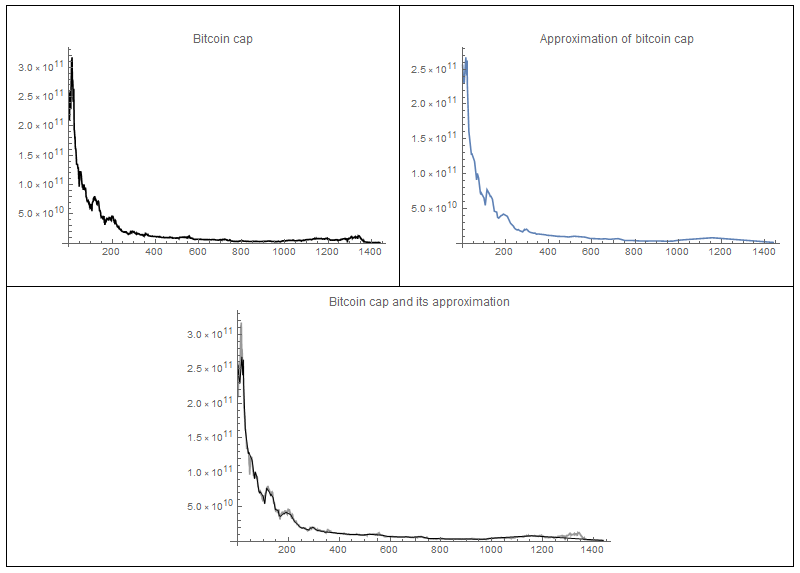}
\caption{Precision of bitcon cap approximation by 100 snapshots with $\EXP{0.1}$ distribution.}
\label{fig:BitcoinCap}
\end{figure}  
We used bitcoin cap data from \cite{bitcoincaps} to test probabilistic snapshots concentrated at the end of the data stream. The data stream consists of 1439 records from 2013-09-30 to 2017-12-31 (one record for one day). Since the length of this stream of data is relatively short, we decided to use only 100 probabilistic snapshots. W used the probability sequence $\alpha_n = \frac{0.1}{\sqrt{n}}$. From Sec. \ref{sec:sublinear} we know that in this case the expected value of snapshots is close to $10\sqrt{1439} \approx 380$.

We added to generated probabilistic snapshots to points: the first one, and the last one.
The results of this experiment is shown at Fig. \ref{fig:BitcoinCap}. Let us remark that the most left points at these diagrams corresponds to data from the day 2017-12-31 and the most right point to 2013-09-30 (so we reversed the typical order of such kind of diagrams). Note that despite the large fluctuations in the bitcoin market at the end of 2017, using only 100 snapshots makes it possible to reproduce the main trend of this parameter of the bitcoins market quite faithfully.